\documentclass[conference]{IEEEtran}
\IEEEoverridecommandlockouts

\usepackage{cite}
\usepackage{amsmath,amssymb,amsfonts,amsthm}
\usepackage{algorithmic}
\usepackage{graphicx}
\usepackage{textcomp}
\usepackage{xcolor}
\usepackage{bbm}
\usepackage{url}

\newtheorem{theorem}{Theorem}
\newtheorem{lemma}{Lemma}
\newtheorem{definition}{Definition}
\newtheorem{remark}{Remark}
\newtheorem{proposition}{Proposition}
\newtheorem{corollary}{Corollary}
\newtheorem{example}{Example}

\def\BibTeX{{\rm B\kern-.05em{\sc i\kern-.025em b}\kern-.08em
    T\kern-.1667em\lower.7ex\hbox{E}\kern-.125emX}}

\begin{document}

\title{On the Capacity of Noisy Frequency-based Channels
\thanks{The research was supported by United States – Israel Binational Science Foundation (NSF-BSF), grant no. 2024763.}}

\author{\IEEEauthorblockN{Yuval Gerzon}
\IEEEauthorblockA{\textit{Technion}\\
gerzon.yuval@campus.technion.ac.il}
\and

\IEEEauthorblockN{Ilan Shomorony}
\IEEEauthorblockA{\textit{UIUC}\\
ilans@illinois.edu}

\and

\IEEEauthorblockN{Nir Weinberger}
\IEEEauthorblockA{\textit{Technion}\\
nirwein@technion.ac.il}
}

\maketitle

\begin{abstract}
We investigate the capacity of noisy frequency-based channels, motivated by DNA data storage in the short-molecule regime, where information is encoded in the frequency of items types rather than their order. The channel output is a histogram formed by random sampling of items, followed by noisy item identification. While the capacity of the noiseless frequency-based channel has been previously addressed,  the effect of identification noise has not been fully characterized. We present a converse bound on the channel capacity that follows from stochastic degradation and the data processing inequality. We then establish an achievable bound, which is based on a Poissonization of the multinomial sampling process, and an analysis of the resulting vector Poisson channel with inter-symbol interference. This analysis refines concentration inequalities for the information density used in Feinstein bound, and explicitly characterizes an additive loss in the mutual information due to identification noise. We apply our results to a  DNA storage channel in the short-molecule regime, and quantify the resulting loss in the scaling of the total number of reliably stored bits.
\end{abstract}

\section{Introduction}
\subsection{Background and Motivation}

DNA storage \cite{church2012next, goldman2013towards, grass2015robust} presents a revolutionary approach to data archiving, and offers unparalleled data density and longevity preserving data for millennia \cite{erlich2017dna, organick2018random, lenz2019anchor, sima2021coding}. While this make DNA storage a prime candidate for archival storage, unlike electronic media, DNA storage lacks inherent addressing, and  data is stored in a pool of molecules. The data retrieval then involves random sampling and sequencing \cite{yazdi2015rewritable, kiah2016codes}. 
In \cite{shomorony2022information}, an information-theoretic model for such channels is presented, called the Noisy Shuffling-Sampling Channel, in which a codeword is represented by $K$ molecules of length $L$ each, stored out of order. As synthesis and sequencing constraints often severely limit the  molecule's length, the molecules are assumed to be of short length, which is captured by the scaling $L=\beta\log K$, where $\beta$ is a constant parameter. The storage rate of such systems is given by the total number of stored bits divided by the total number of nucleotides in the molecules, i.e., $KL$. The result is a fundamental dichotomy regarding the capacity \cite{shomorony2022information}: In the \textit{long-molecule regime}, where $\beta >\frac{1}{\log|\mathcal{A}|}$, capacity is achieved by indexing: reserving a portion of the strand to store an ID, effectively restoring order \cite{shomorony2022information}. In the \textit{short-molecule regime} where $L \leq\frac{\log K}{\log|\mathcal{A}|}$, the overhead of indexing is prohibitive, and the asymptotic capacity vanishes. Formally, the number of maximally reliable stored bits is $o(KL)$. 

Nonetheless, due to its remarkable physical density, the total number of reliably stored bits can still be huge even if the capacity vanishes. A conjecture on the scaling of this  number of bits was put forth in \cite[Sec. 7.3]{shomorony2022information}. In this short-molecule regime, the number of different molecule types is less than $K$, and so each codeword must contain more than a single instance of at least one molecule. In other words, the data is encoded to the frequency of the different molecules in the codeword. This motivated us to propose a general \emph{frequency-based channel} model \cite{gerzon2025capacity}, in which the decoder observes only the \emph{composition} (or histogram) of the stored items (which are the DNA molecules in DNA storage), rather than their ordered sequence. In \cite{gerzon2025capacity} the capacity of this channel was characterized in case that the decoder can randomly sample from the pool of items (molecules in DNA storage), and identify the sampled item without noise. As a corollary, we applied this result to DNA storage, and resolved the conjecture of \cite[Sec. 7.3]{shomorony2022information}, though only for  the regime $\beta \in \left(\frac{1}{2|\log\mathcal{A}|}, \frac{1}{\log|\mathcal{A}|}\right)$. In \cite{tamir2025DNA}, the proof  was completed for the entire regime $\beta \in \left(0, \frac{1}{\log|\mathcal{A}|}\right)$, though using a rather different type of analysis. A related model, based on Poisson sampling, was considered in \cite{bello2024lattice}, and the case in which the frequency-based channel has \textit{unlimited input resolution} was considered in \cite{tamir2025achievable, 11127184},  deriving achievable rates and error probability bounds.

All previous works have assumed that the only randomness of the channel is the result of random sampling of items from the pool, and the identification of the sampled item at the decoder side is noiseless. This models the case of noiseless sequencing in a DNA storage system. However, practical sequencers (e.g., Illumina, Nanopore), as well as other realistic channel models,  introduce significant noise, e.g., \cite{heckel2019characterization}. In the context of frequency-based channels, these physical errors manifest uniquely. Synthesis and sequencing imperfections, such as substitution errors \cite{organick2018random}, effectively ``transport'' mass between histogram bins, causing an item of type $i$ to be misidentified as type $j$, while some noise models, e.g., deletion or erasure, will lead to new counts at the output.

\subsection{Contribution}
In this work, we address the capacity of frequency-based channels in the noisy case. Our main contribution is a characterization of capacity bounds for the noisy frequency-based channels. Specifically, we consider a noise model described as a Markov kernel, defining the probability of identifying at the decoder a symbol of one type whenever the sampled symbol is of different type.  This noise leads to a sort of an Inter-Symbol Interference (ISI) in the frequency domain: the observed count of type $j$ depends on a linear combination of input counts, and under certain noise models, such as Erasure, the type $j$ is not necessarily a member of the original types. We explain why the converse bound in the noiseless case \cite{gerzon2025capacity} is intact in the noisy case via a data-processing argument, and put the main technical effort to derive a new achievability bound. 

We base our analysis technique on \cite{gerzon2025capacity}, as the methods of \cite{tamir2025achievable, 11127184} are more challenging to extend to the noisy case.  The main challenge regarding this adaptation is that various parts of the proof strongly rely on the fact that the number of samples of item of a given type only depends on the number of similar items in the input codeword (pool), and not on the number of other items. This clearly breaks when identification noise is present. As one example, an analysis of the concentration of the information density in the  noiseless case \cite{gerzon2025capacity} based on the simple Hoeffding inequality is replaced for the noisy case with Talagrand inequality, which necessitates establishing a suitable Lipschitz  property for entropy.

Specifically, our lower bound (Theorem \ref{thm:main_noisy_achievability}) establishes an upper bound on the capacity additive degradation due to identification noise, under proper conditions on the scaling of the number of different possible items and total number of items in the codeword. As a corollary, we obtain such a result for the scaling of total number of reliably stored bits in DNA storage channel. Due to the use of more involved arguments, our results apply only for the regime $\beta \in \left(\frac{2}{3|\log\mathcal{A}|}, \frac{1}{|\log\mathcal{A}|}\right)$, which is a minor degradation from \cite{gerzon2025capacity}. Nonetheless, it can be shown that if the size of the output alphabet of the channel is \textit{less} than then input alphabet, different arguments can extend this regime back to the noiseless one of $\beta \in \left(\frac{1}{2|\log\mathcal{A}|}, \frac{1}{|\log\mathcal{A}|}\right)$. For conciseness, we do not formally state this in this paper.

\section{System Model}

For an integer $k \in \mathbb{N}_{+}$, let $[k]:=\{1, 2, \dots, k\}$. For $a, b \in \mathbb{R}$, let $a \wedge b := \min\{a, b\}$. Logarithms and exponents are taken to the natural base, and $h_{b}(\cdot)$ is the binary entropy function. We use standard notation for information-theoretic measures.

We denote the number of distinct input object types by $n$. The input codeword is represented by a count vector $x^n = (x_1, \dots, x_n) \in \mathbb{N}^n$, where $x_i$ is the number of objects of the $i$th type in the pool. The input is subject to a total object constraint $\sum_{i=1}^n x_i = n g_n$, where $g_n$ is the normalized average abundance. We assume $g_n \to \infty$ as $n \to \infty$.

To model noisy identification, we distinguish between the input alphabet size $n$ and the output alphabet size, denoted by $m_n\equiv m$. The channel is characterized by a row-stochastic transition matrix $W^{(n)}\equiv W \in \mathbb{R}^{n \times m}$, where $W_{i,j}$ represents the probability that an input object of type $i\in[n]$ is identified as type $j \in [m]$. The reading process consists of drawing $nr_n$ samples, and the resulting output count vector is denoted by $y^m \in \mathbb{N}^m$.

For reading the codeword, the decoder samples $n r_n$ items from the pool, and the item is identified via the Markov kernel $W^{(n)}$. For an input $x^n$, the output vector $Y^m$ thus follows a Multinomial distribution parameterized by  probabilities $\mathbf{p}=(p_1,\ldots,p_m)$ as
\begin{align}
    Y^{m} \sim \text{Multinomial}\left(nr_{n}, \mathbf{p}\right), \quad p_j = \frac{1}{ng_n}\sum_{i=1}^n x_i W^{(n)}_{i,j}.
\end{align}

The size of the largest code for $n$ object types, normalized total count of input objects $g_{n}$, normalized number of sampled objects $r_{n}$, under a given error probability $\epsilon_{n}\in(0,1)$, and noisy channel $W^{(n)}$ is denoted by $M^{*}_n \equiv M^{*}_n(\epsilon_{n},g_{n},r_{n},W^{(n)})$. The capacity is the maximal value of $\lim_{n\to\infty}M^{*}_n$ for vanishing error probability $\epsilon_n=o(1)$.

\section{Capacity Bounds}

We begin by outlining a  converse bound: 

\begin{theorem}[Converse]
\label{thm:converse}
For any  sequence of Markov kernel channels $\{W^{(n)}\}$, the  noisy frequency-based channel satisfies
\begin{align}
    \frac{1}{n}\log M^*(g_n, r_n, W) \leq \frac{1}{2} \log (r_n \wedge e g_n) + o_n(1).
\end{align}
\end{theorem}

\begin{proof}
As mentioned in \cite[Thm. 1]{gerzon2025capacity}, the bound holds for the noisy case.
The proof relies on stochastic degradation. Consider the noiseless channel where $W^{(n)}=I_n$. The output of the noisy channel $Y^m$ can be simulated by first passing $X^n$ through the noiseless channel to obtain $Y_{ideal}^n$, and then applying the noise matrix $W^{(n)}$ to the samples.
Formally, the noisy channel is a stochastic degradation of the noiseless channel. By the Data Processing Inequality (DPI), the mutual information $I(X^n; Y^m) \leq I(X^n; Y_{ideal}^n)$.
The upper bound for the noiseless case derived in \cite{gerzon2025capacity} using the duality of capacity \cite{lapidoth2003asymptotic, lapidoth2008capacity} applies directly. The term $\frac{1}{2}\log(eg_n)$ arises from the counting constraint ("stars and bars") on the integer inputs, while $\frac{1}{2}\log r_n$ arises from the sampling noise.
\end{proof}

Our main result in this paper is a lower bound on the capacity.

Let $\Psi(\mu) := (\mu+1) h_b(\frac{1}{\mu+1})$ be the maximal entropy of non-negative integer random variable (RV) with a mean upper bounded by $\mu$. 
\begin{definition}[Well-Conditioned Transition Matrix]
\label{def:well_conditioned}
A transition matrix $W$ is well-conditioned with parameters $\tau_n > 0, \eta >0$ and a constant $C_{\text{max}}$ if for all $j$:
\begin{align}
    \kappa_j(W) = \frac{\max_{i} W_{i,j}}{\min_{i:W_{i,j}>0} W_{i,j}} \leq \frac{1}{\tau_n}\\
    n^{-\eta} \le \sum_{i=1}^n W_{ij} \le C_{\text{max}}
\end{align}
\end{definition}
\begin{theorem}[Achievability]
\label{thm:main_noisy_achievability}
Assume $g_n=\omega_n(1)$, $n=\Omega(g_{n}^{1+\zeta})$ for some $\zeta>0$, $r_n = \Theta(g_n)$, and $W$ is a \emph{well-conditioned} matrix (Def. \ref{def:well_conditioned}), and satisfies $-\log \tau_n = o(\sqrt{n})$ then
\begin{align}
\frac{1}{n}\log M^* \geq\;&
\frac{1}{2}\log r_n - \Psi\!\left(\frac{r_n}{g_n}\right)
\notag\\
&\quad + \frac{1}{2n}\log\det(W^{(n)}(W^{(n)})^\top) - o_n(1).
\end{align}
\end{theorem}

While the converse bound remained unchanged from the noiseless case, it may still be useful in an assessment of the tightness of the achievability bound. The achievability bound is altered by an additive term  $\frac{1}{2n}\log \det(W^{(n)}(W^{(n)})^\top)$. Thus, if the sequence of Markov kernels satisfies that this additive term is $o(n)$, the gap is guaranteed to converge to the same gap between the upper and lower bounds in the  noiseless case. For this class of channels, the difference between the noiseless case and the noisy case is asymptotically vanishing. 

Our result shows a robust bound for a large group of error stochastic models with varying error ranges, which shows the great value of the achievability bound for the noisy case of frequency-based channels. We denote the degradation in capacity by $\Delta$ for each channel.

\begin{example} \label{exa:general_substitution}

Consider a frequency-based channel with general symmetric-substitution noise.

The transition matrix $W \in \mathbb{R}^{n \times n}$ is given by:
\begin{align}
    W^{(n)} = \left(1 - p - \frac{p}{n-1}\right)I_n + \frac{p}{n-1}J_n,
\end{align}
where $J_n$ is the all-ones matrix. The eigenvalues of $W^{(n)}$ are $\lambda_1=1$ and $\lambda_{2,\dots,n} = 1 - \frac{np}{n-1}$. Consequently,
\begin{align}
    \det(W^{(n)}(W^{(n)})^\top) = \left(1 - \frac{np}{n-1}\right)^{2(n-1)},
\end{align}
and the capacity penalty term is therefore
\begin{align}
    \Delta_{\text{sub}}
    &= \frac{n-1}{n} \log \left( 1 - \frac{np}{n-1} \right).
\end{align}
Asymptotically, as $n \to \infty$, the term $\frac{np}{n-1} \to p$, and the capacity degradation converges to $\log(1-p)$.

\end{example}

\section{Application to DNA Storage}
In the DNA storage model, the items are molecules of length $L$ from an alphabet $\mathcal{A}$, typically
$\mathcal{A} = \{A,C,G,T\}$. This corresponds to $n = |\mathcal{A}|^L$, and the total number of items is $K=ng_n$, and the total number of sampled molecules is $N=nr_n$. 
For the noise matrix, we first define:
\begin{definition}[A single-nucleotide channel]
    \label{def:single_nucleotide_channel}
    The single-nucleotide channel $w$ is defined as noise kernel of a single DNA letter in a DNA strand.
\end{definition}

We assumed that the single-nucleotide channel operates independently on each of the $L$ nucleotides in the molecule (and independently across molecules). Thus, the transition matrix  is the Kronecker power of the single-nucleotide channel $w$, given by $W = w^{\otimes L}$.

In the DNA storage channel, it is more convenient to index the channel according to the total number of molecules $K$. Then, we let $M_{\text{DNA}}^{*}(L_{K},N_{K},\epsilon_{K},w)$ denote the maximal cardinality of a codebook with error probability $\epsilon_K$.

\begin{corollary}[to Theorem \ref{thm:main_noisy_achievability}]
\label{cor:DNA_thm_cor}
Assume that $L_{K}=\beta\log K$ where $\beta \in \left(\frac{2}{3\log|\mathcal{A}|}, \frac{1}{\log|\mathcal{A}|}\right)$ then Theorem \ref{thm:main_noisy_achievability} implies that 
\begin{align}
&\frac{M_{\text{DNA}}^{*}(L_{K},N_{K},\epsilon_{K},w)}{LK^{\beta\log|{\cal A}|}} \nonumber \\\geq&\frac{1}{2\beta}\frac{\log(N_{K})}{\log K}-\frac{\Psi\left(\frac{N_{K}}{K^{1-\beta\log|{\cal A}|}}\right)}{\beta\log K} + \Delta
+o\left(\frac{1}{\log K}\right) 
\end{align}
where 
\begin{equation}
\Delta := \frac{1}{2|\mathcal{A}|}
 \log\det{ww^\top}<0
\end{equation}
(using $\det(A^{\otimes L}) = (\det A)^{L \cdot a^{L-1}}$ for $A\in\mathbb{R}^{a\times a}$).
\end{corollary}

\begin{example}[Erasure sequencing channel]
Assume that $|\mathcal{A}|=4$. 
Each nucleotide is erased with probability $\epsilon<1/|\mathcal{A}|$ (required for the Well-conditioned transition matrix assumption to holds), yielding alphabet size $q+1=5$.
The matrix is $w = [(1-\epsilon)I_4 \;|\; \underline{\epsilon}]$, and $\tau = 1$ hence $\log \tau = 0 = o\left(\sqrt{n}\right)$. 
Calculating
\begin{align}
    ww^\top = (1-\epsilon)^2 I_4 + \epsilon^2 \mathbf{1}_4 \mathbf{1}_4^\top,
\end{align}
the eigenvalues are $\lambda_1 = (1-\epsilon)^2 + 4\epsilon^2$ and $\lambda_{2,3,4} = (1-\epsilon)^2$.
The penalty term is thus
\begin{align}
    \Delta_{eras} = \frac{1}{8} \left[ \log((1-\epsilon)^2 + 4\epsilon^2) + 6\log(1-\epsilon) \right].
\end{align}
\end{example}

\begin{example}[Substitution sequencing channel]
For the  symmetric substitution channel with error $p$ it holds that 
\begin{align}
    w = \left(1-p - \frac{p}{3}\right)I_4 + \frac{p}{3}J_4.
\end{align}
where $J_4$ is a square all-ones matrix of size $4$. Let us denote $d_p := \left(\frac{1-p}{\frac{p}{3}}\right)\wedge \left(\frac{\frac{p}{3}}{1-p}\right)$, hence $\tau = \left(d_p\right) ^ L$, and so 
\begin{align}
    \log \tau = L \cdot \log \left(d_p\right) \approx \log n \cdot \log \left(d_p\right) = o\left(\sqrt{n}\right)
\end{align}
The eigenvalues are $1,1 - \frac{4p}{3}$ (multiplicity 3).
The penalty is:
\begin{align}
    \Delta_{sub} = \frac{3}{4} \log \left( 1 - \frac{4p}{3} \right).
\end{align}
\end{example}

\section{Proof Outline of Theorem  \ref{thm:main_noisy_achievability}}
\label{sec:proof_main}

The proof follows the proof method of the noiseless case \cite[Thm. 2]{gerzon2025capacity}, and is based on the extended Feinstein bound. Let $P_{Y^{m}\mid X^{n}}$ denote the Markov kernel from a random input $X^{n}$ to the output $Y^{m}$, according to \eqref{def:well_conditioned}. Let $P_{X^{n}}$ denote the input distribution, and $P_{Y^{m}}$ be the corresponding output distribution. Let the information density be $i(x^{n};y^{m}):=\log\frac{P_{Y^{m}\mid X^{n}}(y^{m}\mid x^{n})}{P_{Y^{m}}(y^{m})}$.The extended Feinstein bound   \cite{polyanskiy2024information} assures the following: Let $F_n$ be a constraint set for the codewords (e.g. $\sum_{i=1}^n x_i=ng_n$). For any $\gamma>0$ and $M\in\mathbb{N}_{+}$ there exists a codebook ${\cal C}_{M}=\{x^{n}(1),\ldots,x^{n}(M)\}$ such that $x^{n}(j)\in F_{n}$ for all $j\in[M]$, and whose maximal error probability is $\epsilon_{n}$, where
\begin{equation}
\epsilon_{n}P_{X^{n}}(F_{n})\leq\Pr\left[i(X^{n};Y^{m})\leq\log\gamma\right]+\frac{M}{\gamma}.
\end{equation}

Accordingly, the proof of the noiseless case \cite[Thm. 2]{gerzon2025capacity} is based on three, tightly intertwined, main steps: (i) Analyzing the concentration of the information density (\cite[Prop. 5]{gerzon2025capacity}), (ii) analyzing the expected information density (which is the mutual information) (\cite[Prop. 7]{gerzon2025capacity}) and the effect of integer inputs on its value, and  (iii) analyzing the probability of the constraint set $P_{X^{n}}(F_{n})$ \cite[Prop. 9]{gerzon2025capacity}. The last step remains unchanged, since the probability of the typical input set depends solely on the encoder's distribution choice and holds independently of the channel matrix. We thus henceforth only focus on the first two steps. 

In general, analyzing the multinomial channel is complicated due to the statistical dependency between the output coordinates (reflected, e.g., in the constraint $\sum Y_j = nr_n$). Thus, as in \cite{gerzon2025capacity}, we utilize the "Poissonization" technique \cite{mitzenmacher2017probability} to replace  the analysis of the multinomial  channel with a Poisson channel  $P_{Z^m|X^n}$ where
\begin{align}
    Z_j \sim \text{Poisson}\left(\lambda_j(x^n)\right), \quad \lambda_j(x^n) = \frac{r_n}{g_n} \sum_{i=1}^n x_i W_{i,j}.
    \label{eq:poisson_model}
\end{align}
Here, $\lambda_j$ is the expected count of type $j$. Note that, unlike the noiseless case, where $Z_i$ depends only on $x_i$, here $Z_j$ depends on the entire vector $x^n$ via the $j$th column of $W$. Following this reduction, the proof continues with the analysis of the Poisson channel. As said, two steps are required. First, to show that the information density in the Poisson channel $ i(x^n; z^m) := \log \frac{P_{Z^m|X^n}(z^m|x^n)}{P_{Z^m}(z^m)}$ concentrates, and in fact, due to terms that are the result of the Poissonization, it is required to show that this concentration is sufficiently fast. Second, to analyze the concentrated value, which is the expected value of the information density, given by the mutual information $I(X^n;Z^m)$. In these two steps lie the main technical challenges of the analysis of the noisy case, compared to the noiseless case. In the first step, proving concentration requires significantly more complicated analysis, and in the second step, the noise in the channel leads to a reduction in the mutual information compared to the noiseless case, which decreases the achieavable rate. In what follows, we will outline the main ideas in these derivations.

\subsection{Concentration of the Information Density} 

In this section, we consider the concentration of $i(X^n; Z^m)$ to its expected value, under the channel \eqref{eq:poisson_model}. This is achieved by restricting the input counts to a support $X_i\in[s_n]$ and is done in two sub-steps. First, conditioning on $X^n$ and analyzing the concentration with respect to (w.r.t.) $Z^m$ to its conditional expectation (a conditional mutual information), and second, analyzing the convergence of the conditional expectation for random $X^n$. In the noiseless case, the first step is derived by a concentration inequality for Lipschitz functions of Poisson RVs due to Bobkov-Ledoux \cite{bobkov1998modified}, and the second step via the standard Hoeffding inequality. Here, for the noisy case, both steps are more challenging. 

For the first step, bounding the Lipschitz semi-norm governing the concentration of the information density is more complicated, since the identification noise $W$ introduces dependencies between input intensities  and output counts. This requires to incorporate the condition number of $W$ into the bound on the Lipschitz semi-norm. Concretely, we evaluate the  Lipschitz semi-norm for the information density $f_{x^n}(z^m) := i(x^n; z^m)$ in the noisy setting, generalizing \cite[Lem. 10]{gerzon2025capacity}. Since this function has discrete inputs, the Lipschitz semi-norm is defined by the discrete derivative, 
\begin{align}
\beta_{\text{Lip}} := \sup_{z^m \in \mathbb{N}^m, j \in [m]} \left| f_{x^n}(z^m + e_j) - f_{x^n}(z^m) \right|,\\
\lambda_{\text{Lip}} = \max_j \lambda_j
\end{align}
where $e_j$ is the $j$th standard basis vector. The key difference from the noiseless case is that each output component $Z_j$ depends on the weighted sum $\sum_{i=1}^n x_i W_{i,j}$ rather than just a single input component. Using the Poisson probability mass function,  we show using some algebra that
\begin{align}
\frac{\min_{i,j: W_{i,j}>0} W_{i,j}}{s_n \max_{i,j} W_{i,j}}&\leq
\frac{P_{Z_j|X^n}(z_j + 1 | x^n)}{P_{Z_j|X^n}(z_j | x^n)} \cdot \frac{P_{Z_j}(z_j)}{P_{Z_j}(z_j + 1)} \\ &\leq s_n \cdot \frac{\max_{i,j} W_{i,j}}{\min_{i,j: W_{i,j}>0} W_{i,j}},
\end{align}
where $s_n$ is a support size we impose, as a part of the proof, on the probability distribution of each $X_i$. By an analogous argument for 
$\frac{P_Z^m(z^m)}{P_Z\bigl(z^m + e_j\bigr)}$, taking logarithms, using the definition of the Lipschitz semi-norm and the well-conditioning property (Definition \ref{def:well_conditioned}), we obtain 
\begin{align}
\beta_{\text{Lip}} &\leq \log s_n + \log\left(\frac{\max_{i,j} W_{i,j}}{\min_{i,j: W_{i,j}>0} W_{i,j}}\right) \\
&\leq \log s_n - \log \tau_n.
\end{align}

In order to obtain sufficiently fast concentration, it is required that $\beta_{\text{Lip}} = o(\sqrt{n})$, which constrains both $s_n$, as a technical part of the proof, and 
also $\log \tau = o(\sqrt{n})$. The later is an additional constraint, appearing only for the noisy case. 

For the second step, to wit, the concentration over a random input $X^n$, Hoeffding's inequality cannot be used since the involved RVs are dependent. Instead, we invoke Talagrand's concentration inequality for convex functions \cite[Corollary 4.23]{van2014probability}): If $f: [0,1]^n \to \mathbb{R}$ is convex and $\|\nabla f\|_2 \le L$, then $f$ is $L^2$-sub-Gaussian. Concretely, the relevant function is  
\begin{equation}
h(x^{n}):=\sum_{j=1}^{m}H_{\text{Poiss}}\left(\frac{r_{n}}{g_{n}}\sum_{i=1}^{n}W_{ij}x_{i}\right), 
\end{equation}
where $H_{\text{Poiss}}(\lambda)$ is the entropy of a Poisson RV with parameter $\lambda$. To use Talagrand's inequality, we need to bound its gradient norm $\|\nabla h\|_2$, and show that it is a convex function. We achieve that by analyzing the second derivative of $H_{\text{Poiss}}$. This is non trivial, there are no closed-form expressions for the Poisson entropy. Nonetheless, we utilize the Poisson forward-difference identity \cite{adell2010sharp} to show that 
\begin{equation}
    \left|\frac{\mathrm{d}}{\mathrm{d}\lambda}H'(\lambda)\right| \leq \log(1+1/\lambda), 
\end{equation}
and then again to show that 
\begin{equation}
\frac{\mathrm{d}^{2}}{\mathrm{d}\lambda^{2}}H_{\mathrm{Poiss}}(\lambda) \le -\frac{e^{-\lambda}}{\lambda} < 0.
\end{equation}
Hence, $\lambda\to H_{\text{Poiss}}(\lambda)$ is a convex function, and so is $f(x^n)$. Similarly, the gradient norm is bounded by the first derivative, where in this case, we account for the column sums of $W$. Assuming the column sums satisfy $\sum_{i=1}^n W_{ij} \geq n^{-\eta}$ for some $\eta > 0$ (ensured by the well-conditioning assumption of Definition \ref{def:well_conditioned}), we have
\begin{align}
    \lambda_j(x^n) = \frac{r_n}{g_n}\sum_{i=1}^n x_i W_{ij} \geq \frac{r_n}{g_n} n^{-\eta}.
\end{align}
For sufficiently large $n$, this gives:
\begin{align}
    |H'(\lambda_j)| \leq \log(1 + n^{\eta}) \approx \eta \log n.
\end{align}
Substituting into the gradient expression and using the row-stochastic property $\sum_{j=1}^m W_{k,j} = 1$:
\begin{align}
    \left|\frac{\partial h(x^n)}{\partial x_k}\right| &\leq \frac{r_n}{g_n} \sum_{j=1}^m W_{k,j} \cdot \eta \log n \\
    &= \frac{r_n}{g_n} \eta \log n \sum_{j=1}^m W_{k,j}  \\
    &= \frac{r_n}{g_n} \eta \log n.
\end{align}

The squared Euclidean norm of the gradient becomes:
\begin{align}
    \|\nabla h\|_2^2 &= \sum_{k=1}^n \left(\frac{\partial h(x^n)}{\partial x_k}\right)^2 \\
    &\leq \sum_{k=1}^n \left(\frac{r_n}{g_n} \eta \log n\right)^2 
    = n \left(\frac{r_n}{g_n} \eta \log n\right)^2.
\end{align}

Now, for Talagrand's inequality to yield a vanishing probability bound, the variance proxy must satisfy $\sigma^2 = o(n^2)$. The variance proxy is proportional to $(s_n)^2 \|\nabla h\|_2^2$. This leads to the requirement:
\begin{align}
    (s_n)^2 \cdot n \left(\frac{r_n}{g_n}\right)^2 \log^2 n = o(n^2).
\end{align}

Since $\frac{r_n}{g_n}$ is bounded by a constant (assumption: $cg_n \leq r_n \leq eg_n$), this simplifies to $s_n = o(\sqrt{n})$. 

To conclude this derivation, we mention that the use of Talagrand's inequality is the source of the condition on $\beta$, when the theorem is applied to a DNA storage channel. Specifically, in the proof we choose $s_n=g_n^{1+\rho}$ for some small value of $\rho$. Choosing the DNA parameters $n=|\mathcal{A}|^{\beta\log K}$ and $ng_n=K$, the above leads to $s_n \approx n^{\frac{1}{\beta \log |\mathcal{A}|} - 1}$. To assure that $s_n=o(\sqrt{n})$ it is thus required that $\beta > \frac{2}{3 \log |\mathcal{A}|}$. Therefore, one possible way to increase this regime is to either use a different concentration inequality than Talagrand's inequality, or to refine the analysis. 
\subsection{Characterizing the Capacity Penalty}
Finally, we characterize the loss in mutual information due to the $W$. This extends \cite[Prop. 7]{gerzon2025capacity} by quantifying the entropy loss induced by the noise in the channel. We decompose $I(X^n;Z^m)=H(Z^m)-H(Z^m\mid X^n)$, and separately analyze the two entropy terms. For the conditional entropy $H(Z^m\mid X^n)$, we note that conditioned on $X^n$, the output is a sequence of independent Poisson RVs (see \eqref{eq:poisson_model}), and so the conditional entropy is additive. For the output entropy $H(Z^m)$, a direct evaluation appears cumbersome, since $Z^m$ is a vector of dependent RVs, mixtures of Poisson distributions. To overcome this, we follow the analysis of the Poisson channel capacity \cite[Proposition 11]{lapidoth2008capacity}. There,  the output entropy of the Poisson channel is lower bounded with the \textit{differential entropy} of its input, which in our case is $h(Wx)$. In turn, the differential entropy of a linear transformation of a random vector is directly related to the differential entropy of the random vector: the differential entropy $h(WX)$ is larger by $\frac{1}{2}\log \det(WW^\top)$ from the differential entropy $h(X)$ (if the dimension of the output space is at least that of the input ($m \ge n$), and $W$ is full rank, then this follows immediately from \cite[Theorem 8.6.4]{cover1991elements}, based on the  determinant of Jacobian matrix of the transformation). In turn, all the above leads to 
\begin{align}
    I(X^n; Z^m) \approx I(X^n; Z^n_{\text{ideal}}) + \frac{1}{2}\log \det(WW^\top),
\end{align}
where $Z_{\text{ideal}}^n$ is the output of the Poisson channel when $W$ is the identity matrix (noiseless channel).
With this, following the lines of the proof of the noiseless capacity derived in \cite{gerzon2025capacity}, the achievable rate is lower bounded as stated in Theorem \ref{thm:main_noisy_achievability}. We also mention, that the required scaling of the input support to an effective support $s_n^* = s_n \cdot \max_{j \in [m]} \left( \sum_{i=1}^n W_{ij} \right)$ is also adjusted in the proof compared to the noiseless case, to account for column-sum variations. To control these variations, we require a constant $C_{\text{max}}$ such that $s_n^* \le C_{\text{max}} s_n$, thereby ensuring that the concentration inequality derived via the Bobkov-Ledoux inequality remains valid. This completes the proof.

\section{Conclusion}
We derived an achievability bound on the capacity  of noisy frequency-based channels. The result generalizes the noiseless capacity by adding a penalty term dependent on the Markov kernel matrix $W^{(n)}$. The analysis highlights the impact of the dependency of the output counts of one of the items (molecules in DNA storage) on a set of counts of the input, and possibly all possible input counts. When utilized for DNA storage channel, this results a loss of at most $LK^{\beta\log|{\cal A}|}\cdot \Delta$ in total number of reliably stored bits, where $\Delta$ is a constant depending on $W$. 

Future work,  includes extending this to  other noise models, such as insertion/deletion errors, which break the tensor structure of $W$, and refining the analysis to avoid the $\beta \in \left(\frac{2}{3\log|\mathcal{A}|}, \frac{1}{\log|\mathcal{A}|}\right)$ constraint required for the results to hold for DNA storage channel, as well as a refinement for the limitations over $W$.

\clearpage
\bibliographystyle{IEEEtran}
\bibliography{DNA_short}

\clearpage
\onecolumn
\appendices
\section{Full Proof of Theorem \ref{thm:main_noisy_achievability}}
\subsection{Full Proof of \ref{thm:main_noisy_achievability}}

To establish the achievability bound for the general noisy frequency-based channel, we address the mixing effect of the transition matrix $W$ by adapting the concentration of measure arguments. We utilize a refined Feinstein bound that incorporates the condition number of $W$ into the Lipschitz constant for the information density (Proposition~\ref{prop:Feinstein_Adapted}). Furthermore, we lower bound the mutual information by accounting for the entropy loss induced by the linear transformation $W$ (Theorem~\ref{thm:noisy_achievability}).

\begin{proof}
The proof combines the extended Feinstein bound with the asymptotic analysis of the mutual information in the Poisson domain.

First, we invoke the adapted Feinstein bound (Proposition~\ref{prop:Feinstein_Adapted}). By Poissonizing the multinomial channel, we analyze the information density $i(X^n; Z^m)$ of the channel $Z^m \sim \text{Poisson}(\frac{r_n}{g_n}X^nW)$. Since $W$ is well-conditioned, the information density satisfies the Lipschitz property with semi-norm $\beta = \log s_n - \log \tau$. Using the adapted concentration inequality (Lemma~\ref{prop:Concentration_Under}) and the concentration of conditional entropy (Lemma~\ref{lem:Concentration_of_Poisson_entropy}), we guarantee that for a codebook of size $M$, the error probability $\epsilon_n$ vanishes as $n \to \infty$, provided that $\log M \approx nI(X^n; Z^m)$.

Second, we lower bound the mutual information $I(X^n; Z^m)$. Applying Theorem~\ref{thm:noisy_achievability}, we account for the entropy loss induced by the linear transformation $W$ relative to the noiseless case. This yields the penalty term involving $\det(WW^\top)$, resulting in the bound:
\begin{align}
I(X^n; Z^m) \geq \frac{n}{2}\log r_n - n\Psi\left(\frac{r_n}{g_n}\right) + \frac{1}{2}\log\det(WW^\top) - o_n(1).
\end{align}

Finally, \cite[Proposition~9]{gerzon2025capacity} (which remains valid under the total mean constraint) ensures that there exists a constant composition set $F_n$ with sufficient probability mass, $P_{X^n}(F_n) \geq \frac{1}{3ng_n}$. Combining the mutual information estimate with the error probability bound from the Feinstein argument yields the stated achievable rate.
\end{proof}

\subsection{Adjusting Prop. 5 from \cite{gerzon2025capacity}}

We now combine the concentration of the information density (Lemma~\ref{prop:Concentration_Under}, adapted version of \cite[Lemma 10]{gerzon2025capacity}) and the concentration of the conditional entropy (Lemma~\ref{lem:Concentration_of_Poisson_entropy}, adapted version of Lemma 11 in \cite{gerzon2025capacity}) to establish the achievability bound.

\begin{proposition}[Adapted Proposition 5: Feinstein Bound for Poisson Channels]
\label{prop:Feinstein_Adapted}
Consider a channel $P_{Z^m|X^n}$  where $m \geq n$ and an input distribution $P_{X^n}$. Let $i(x^n; z^m)$ denote the information density:
\begin{align}
i(x^n; z^m) := \log \frac{P_{Z^m|X^n}(z^m|x^n)}{P_{Z^m}(z^m)}.
\end{align}
For any threshold $\gamma > 0$ and blocklength $n$, there exists a code with $M$ codewords and average error probability $\epsilon_n$ satisfying:
\begin{align}
\epsilon_n \leq \mathbb{P}\left[ i(X^n; Z^m) \leq \log \gamma \right] + \frac{M}{\gamma}.
\end{align}
The code size $M$ is chosen as:
\begin{align}
\log M = nI(X^n; Z^n) - 3n\delta_n - \frac{1}{2} \log(6\pi n r_n), 
\end{align}
where $I(X^n; Z^n)$ is the mutual information over the noisy channel $W$.

The maximal error probability $\epsilon_n$ on the multinomial channel from $X^n$ to $Z^m$ is bounded as:

\begin{align}
\epsilon_n \leq \frac{11}{P_{X}^{\otimes n}(F_n)}
\left[ \sqrt{nr_n} \exp \left[ -n\delta_n^2 \cdot \left( \frac{c\eta}{\left(\frac{r_n}{g_n}\right)^2 s_n^2 \log^2 n} \wedge \frac{g_n}{19 C_{\text{max}}r_n s_n (\log s_n - \log \tau)^2} \right) \right] + e^{-n\delta_n} \right].
\end{align}
\end{proposition}

\begin{proof}
The proof follows the standard Feinstein bound approach using the threshold decoder. The error probability is bounded by the probability that the random codeword does not behave typically (concentration of information density) plus the probability that another codeword looks typical (packing lemma).

The exponential decay term is determined by the slower of the two concentration events:
\begin{enumerate}
    \item \emph{Concentration of Information Density given $X^n$:} From Lemma~\ref{prop:Concentration_Under}, the rate is determined by the Lipschitz constant $\beta = \log s_n - \log \tau$, yielding the term $\frac{g_n}{19 C_{\text{max}} r_n s_n (\log s_n - \log \tau)^2}$. this replaces the step in \cite[equation (105)]{gerzon2025capacity}
    \item \emph{Concentration of Expected Information Density (Entropy):} From Lemma~\ref{lem:Concentration_of_Poisson_entropy} (Eq.~\ref{eq:concentration_poisson_min_prob}), the rate is determined by the gradient bound of the Poisson entropy, yielding the term $\frac{c\eta}{(r_n/g_n)^2 s_n^2 \log^2 n}$. this replaces the step in \cite[equation (103)]{gerzon2025capacity}
\end{enumerate}
The rest of equations (100) to (106) in \cite{gerzon2025capacity} require only minor adaptation.
Taking the minimum of these two rates in the exponent provides the valid upper bound for the intersection of these events, which results in
\begin{align}
\epsilon_n P_{X}^{\otimes n} (F_n) \leq 4e\sqrt{nr_n} \exp \left[ -n\delta_n^2 \cdot \left( \underbrace{\frac{c\eta}{\left(\frac{r_n}{g_n}\right)^2 s_n^2 \log^2 n}}_{\text{Entropy Conc. (Lemma \ref{lem:Concentration_of_Poisson_entropy})}} \wedge \underbrace{\frac{g_n}{19 C_{\text{max}} r_n s_n (\log s_n - \log \tau)^2}}_{\text{Info Density Conc. (Lemma \ref{prop:Concentration_Under})}} \right) \right] \nonumber \\
+ \frac{M}{e^{I(X^n;Z^m)-2n\delta_n-\frac{1}{2} \log(6\pi nr_n)}}. \label{eq:final_noisy_bound}
\end{align}.
\end{proof}

\section{Analysis of Adjustments}

\subsection{Degradation of Concentration Parameters}
In the noiseless case (where $W$ is the identity), the Lipschitz semi-norm $\beta$ used in the Bobkov-Ledoux bound was simply $\log s_n$. For the noisy case, as derived in equation \eqref{eq:additive_W_ratio}, the semi-norm increases to:
\begin{align}
\beta_{\text{noisy}} = \log s_n + \log\left(\frac{\max_{i,j} W_{i,j}}{\min_{i,j: W_{i,j}>0} W_{i,j}}\right).
\end{align}
Using the well-conditioning parameter $\tau$, this becomes $\beta_{\text{noisy}} = \log s_n - \log \tau$.
Impact: Since $0 < \tau \le 1$, $-\log \tau \ge 0$. The Lipschitz constant increases, which appears in the denominator of the exponent in the concentration bound (Term A). This implies that for noisy channels, the convergence of the information density to its mean is slower than in the noiseless case, requiring larger block lengths $n$ to achieve the same error probability.

\subsection{Requirements on Matrix $W$ - Well-Defined Matrix}
\label{sec:requirements_on_W}
To ensure the proof holds and $\beta_{\text{noisy}}$ does not diverge (which would render the bound useless), $W$ must satisfy the \emph{Well-Conditioned} property. Specifically:
\begin{enumerate}
    \item Non-vanishing entries (Condition number): The ratio between the maximum and minimum non-zero entries in any column must be bounded. We require $\kappa(W) \leq 1/\tau < \infty$. This ensures $\beta_{\text{noisy}}$ is finite.
    \item $\log \tau = o(\sqrt{n})$, to prevent $\beta_{\text{noisy}}$ from diverging.
    \item Column Sums (Gradient bound): For the tighter bound in Lemma 11 (Eq. \ref{eq:concentration_poisson_min_prob}) to apply, the column sums must not vanish too quickly (bounded below by $n^{-\eta}$). This ensures the gradient of the entropy in the Talagrand inequality does not explode.
    \item Column Sums (maximal Poisson Mean): we wish the parameter $\bar\lambda$ to remain useful for the Bobkov-Ledoux inequality, hence we demand $\sum_{i=1}^n W_{ij} \le C_{\text{max}}$.
\end{enumerate}
If $W$ contains entries arbitrarily close to 0 (without being exactly 0), $\beta \to \infty$, and the concentration guarantee vanishes.

\subsection{Concentration of the Information Spectrum}
Due to the concentration bound of Lipschitz functions of Poisson random variables (RVs) due to Bobkov and Ledoux, the support will remain unchanged, as shown next.

We begin with the first step, for which we recall that $\text{supp}(P_X) \subseteq [s_n] = \{1, 2, \ldots, s_n\}$, and specifically, that $P_X(0) = 0$. We use this assumption to establish that for any $x^n \in \text{supp}(P_X^{\otimes n})$, the random variable
\begin{align}
f_{x^n}(Z^m) := \left[- \log P_{Z^m}(Z^m) + \sum_{j=1}^m \log P_{Z_j|X^n}(Z_j | X^n = x^n)\right] 
\end{align}
concentrates fast around its expected value
\begin{align}
I(Z^m; X^n = x^n) := - \mathbb{E}\left[\log P_{Z^m}(Z^m)\right] + \sum_{j=1}^m \mathbb{E}\left[\log P_{Z_j|X^n}(Z_j | x^n) \bigg| X^n = x^n\right]. 
\end{align}

We achieve this using a concentration bound of Lipschitz functions of Poisson RVs due to Bobkov and Ledoux~\cite[Prop. 11]{bobkov1998modified} stated there in Lemma~27 (Appendix~D). The result is as follows.

We will now adapt \cite[Lemma 10]{gerzon2025capacity}

\begin{lemma}[Adapted Lemma 10]
\label{lemma:Adapted_Lemma_10}
Assume that $\text{supp}(P_X) \subseteq [s_n]$ for some $s_n \in \mathbb{N}_+$. Let $x^n \in ([s_n])^{\otimes n}$. Then, with  $C$ as some constant such that $\bar\lambda
 = \frac{r_n}{g_n}\,s_n C$ for any $\delta \in (0, \frac{g_n}{r_n} s_n)$,
\begin{align}
\mathbb{P}\left[f_{x^n}(Z^m) < I(Z^m; X^n = x^n) - n\delta \bigg| X^n = x^n\right] \leq \exp\left[-n\frac{g_n \delta^2}{19 C r_n s_n \log^2 \left( s_n\cdot\left(\frac{\max_{i,j} W_{i,j}}{\min_{i,j: W_{i,j}>0} W_{i,j}}\right)\right)}\right]. 
\end{align}
\end{lemma}

\begin{proof}
To establish this, we begin by showing that if $x^n \in ([s_n])^{\otimes n}$ then $f_{x^n}(z^m)$ is Lipschitz with semi-norm $\beta = \log s_n$, as follows. We denote by $e^{(j)}(i) = (0, 0, \ldots, 1, 0, \ldots)$ the $i$-th standard basis vector in $\mathbb{R}^m$, similarly to~\cite{gerzon2025capacity}. Let $Z^m | X^n = x^n \sim \text{Poisson}\left(\frac{r_n}{g_n} \sum_{i=1}^n x_i W_{i,j}\right)$ for each component $j$. Then, it holds for any $x^n \in \mathbb{R}_+^n$ and $z \in \mathbb{N}$, yet also $\sum_{i=1}^n x_i W_{i,j}\in \mathbb{R}_+^m$ since $W$ satisfies that all entries are nonnegative values bounded by $1$.

\begin{align}
\frac{P_{Z_j|X^n}(z + 1 | x^n)}{P_{Z_j|X^n}(z | x^n)} &= \frac{e^{-\frac{r_n}{g_n}\sum_{i=1}^n x_i W_{i,j}} \left(\frac{r_n}{g_n}\sum_{i=1}^n x_i W_{i,j}\right)^{z+1}}{(z + 1)!} \cdot \frac{z!}{e^{-\frac{r_n}{g_n}\sum_{i=1}^n x_i W_{i,j}} \left(\frac{r_n}{g_n}\sum_{i=1}^n x_i W_{i,j}\right)^z} \\
&= \frac{r_n \sum_{i=1}^n x_i W_{i,j}}{g_n (z + 1)}. 
\end{align}

This remains unchanged even though $x$ is not an integer input, due to the Poisson distribution that leaves $z$ as a nonnegative integer.

Let $P_{Z_j}$ be the marginal resulting from $P_X \otimes P_{Z_j|X^n}$. Then, similarly,
\begin{align}
\frac{P_{Z_j}(z)}{P_{Z_j}(z + 1)} &= \frac{\sum_{\tilde{x}\in\text{supp}(P_X)} P_X(\tilde{x})P_{Z_j|X^n}(z | \sum_{i=1}^n \tilde{x}_i W_{i,j})}{\sum_{\tilde{x}\in\text{supp}(P_X)} P_X(\tilde{x})P_{Z_j|X^n}(z + 1 | \sum_{i=1}^n \tilde{x}_i W_{i,j})} \\
&\leq \max_{\tilde{x}\in\text{supp}(P_X)} \frac{P_{Z_j|X^n}(z | \sum_{i=1}^n \tilde{x}_i W_{i,j})}{P_{Z_j|X^n}(z + 1 | \sum_{i=1}^n \tilde{x}_i W_{i,j})} \\
&= \max_{\tilde{x}\in\text{supp}(P_X)} \frac{g_n (z + 1)}{r_n \sum_{i=1}^n \tilde{x}_i W_{i,j}}. 
\end{align}

The key difference from the noiseless case is that each output component $Z_j$ now depends on the weighted sum $\sum_{i=1}^n x_i W_{i,j}$ rather than just $x_j$. This affects the Lipschitz constant, which now depends on the range of values that $\sum_{i=1}^n \tilde{x}_i W_{i,j}$ can take over the support of $P_X$. The support of $\sum_{i=1}^n x_i W_{i,j}$ is still discrete, but not integer valued.

The overall change will be\\
\begin{align}
(*) = \frac{P_{Z^m|X^n}(z + e^{(j)}(i) | x^n)}{P_{Z|X^n}(z | x^n)} = \\ \prod_{j=1}^{m} \frac{P_{Z_j|X^n}(z_j + e^{(j)}(i) | x^n)}{P_{Z_j|X^n}(z_j | x^n)} =\\ \frac{P_{Z_i|X^n}(z_i + e^{(i)}(i) | x^n)}{P_{Z_i|X^n}(z_i | x^n)} \cdot \prod_{j=1, j\ne i}^{m} \frac{P_{Z_j|X^n}(z_j | x^n)}{P_{Z_j|X^n}(z_j | x^n)} = \\ \frac{P_{Z_i|X^n}(z_i + e^{(i)}(i) | x^n)}{P_{Z_i|X^n}(z_i | x^n)}
\end{align}

\begin{align}
(**)\quad
\frac{P_Z^m(z^m)}{P_Z\bigl(z^m + e^{(j)}\bigr)}
&= \frac{\displaystyle \sum_{x^n}P_{X^n}(x^n)\,P_{Z^m\mid X^n}\bigl(z^m\mid x^n\bigr)}{\displaystyle \sum_{x^n}P_{X^n}(x^n)\,P_{Z^m\mid X^n}\bigl(z^m + e^{(j)}\mid x^n\bigr)} \\[6pt]
&\le \max_{x^n} \frac{P_{Z^m\mid X^n}(z^m\mid x^n)}{P_{Z^m\mid X^n}(z^m + e^{(j)}\mid x^n)} \\[6pt]
&= \max_{x^n} \frac{P_{Z_j\mid X^n}(z_j\mid x^n)\,\prod_{k\neq j}P_{Z_k\mid X^n}(z_k\mid x^n)}{P_{Z_j\mid X^n}(z_j+1\mid x^n)\,\prod_{k\neq j}P_{Z_k\mid X^n}(z_k\mid x^n)} \\[6pt]
&= \max_{x^n} \frac{P_{Z_j\mid X^n}(z_j\mid x^n)}{P_{Z_j\mid X^n}(z_j+1\mid x^n)} = \max_{x^n} \frac{g_n\,(z_j+1)}{r_n\sum_{t=1}^n x_t\,W_{t,j}}.
\end{align}

Now, we wish to work with an unknown set of discrete values ranging from $0$ to $n\cdot s_n \cdot \max_{j} \sum_{i=1}^n W_{i,j}$. Hence,
\begin{align}
(*) \cdot (**) &= \frac{P_{Z_j|X^n}(z + 1 | x^n)}{P_{Z_j|X^n}(z | x^n)} \cdot \frac{P_{Z_j}(z)}{P_{Z_j}(z + 1)} \leq \max_{\tilde{x}\in\text{supp}(P_X)} \frac{\sum_{i=1}^n \tilde{x}_i W_{i,j}}{\sum_{i=1}^n x_i W_{i,j}} \leq s_n \cdot \frac{\max_{i,j} W_{i,j}}{\min_{i,j} W_{i,j}}. 
\end{align}

For a doubly stochastic or well-conditioned transition matrix, this ratio is bounded, giving us the Lipschitz property with semi-norm $\beta = \log s_n + \log\left(\frac{\max_{i,j} W_{i,j}}{\min_{i,j} W_{i,j}}\right)$.

\begin{align}
\frac{P_{Z_j|X^n}(z + 1 | x^n)}{P_{Z_j|X^n}(z | x^n)} \cdot \frac{P_{Z_j}(z)}{P_{Z_j}(z + 1)} &\geq \min_{\tilde{x}\in\text{supp}(P_X)} \frac{\sum_{i=1}^n \tilde{x}_i W_{i,j}}{\sum_{i=1}^n x_i W_{i,j}} \geq \frac{\min_{i,j: W_{i,j}>0} W_{i,j}}{\max_{i,j} W_{i,j}} \geq \frac{1}{s_n}. 
\end{align}

Thus, for any $x^n \in \text{supp}(P_X) \subseteq [s_n]^n$ and $z \in \mathbb{N}$,
\begin{align}
\left|\log \frac{P_{Z_j|X^n}(z + 1 | x^n)}{P_{Z_j}(z + 1)} - \log \frac{P_{Z_j|X^n}(z | x^n)}{P_{Z_j}(z)}\right| \leq \log s_n + \log\left(\frac{\max_{i,j} W_{i,j}}{\min_{i,j: W_{i,j}>0} W_{i,j}}\right). 
\end{align}

The additive form of $f_{x^n}(z^m)$ then implies that
\begin{align}
\label{eq:additive_W_ratio}
\max_{z^m\in\mathbb{N}^m} |f_{x^n}(z^m + e^{(j)}(i)) - f_{x^n}(z^m)| \leq \log s_n + \log\left(\frac{\max_{i,j} W_{i,j}}{\min_{i,j: W_{i,j}>0} W_{i,j}}\right). 
\end{align}
\end{proof}

\begin{lemma}[Concentration Under Well-Conditioning]
\label{prop:Concentration_Under}
If $W$ is well-conditioned (definition \ref{def:well_conditioned}) with parameters $\tau_n > 0, \eta >0$ and a constant $C_{\text{max}}$, then the concentration inequality in Lemma~10 holds with Lipschitz constant $\beta = \log s_n - \log \tau$, and the probability bound becomes:
\begin{align}
\mathbb{P}\left[f_{x^n}(Z^m) < I(Z^m; X^n = x^n) - n\delta \bigg| X^n = x^n\right] \leq \exp\left[-n\frac{g_n \delta^2}{19 C_{\text{max}} r_n s_n (\log s_n - \log \tau)^2}\right].
\end{align}
\end{lemma}

\begin{proof}
    The proof is imidiate when using Lemma \ref{lemma:Adapted_Lemma_10} with definition \ref{def:well_conditioned}.
\end{proof}

Assumption for the Remainder of this Work: Throughout this analysis, we assume that the transition matrix $W$ is well-conditioned with some fixed parameters $\tau_n > 0, \eta >0$ and a constant $C_{\text{max}}$. This ensures that all concentration bounds remain meaningful and that the achievability proof maintains its probabilistic guarantees.

Now, we will adjust \cite[Lemma~27]{gerzon2025capacity}.

\subsubsection*{Instantiation of the Bobkov-Ledoux Parameters}

In our noisy-$W$ setting, \cite[Lemma~27]{gerzon2025capacity} (Bobkov-Ledoux) applies to $f_{x^n}(Z^m)$ with the following choices:

(i) Poisson scale:
We must define a new maximal effective support bound, $s_n^*$, that bounds the parameter for all $j$:
    \begin{align}
        s_n^* = s_n \cdot \max_{j \in [m]} \left( \sum_{i=1}^n W_{ij} \right).
    \end{align}
\begin{align}
\bar\lambda
&=\max_{j=1,\dots,m}\lambda_j
= \max_{j}\Bigl(\frac{r_n}{g_n}\sum_{i=1}^n x_iW_{i,j}\Bigr)
= \frac{r_n}{g_n}\,s_n^*\leq \frac{r_n}{g_n}\,s_n C_{\text{max}},
\end{align}

(ii) Lipschitz semi-norm:
\begin{align}
\beta
&=\max_{j,z^m}
\bigl|f_{x^n}(z^m+e^{(j)})-f_{x^n}(z^m)\bigr|
\le \log s_n + \log\frac{\max_{i,j}W_{i,j}}{\min_{i,j:W_{i,j}>0}W_{i,j}}
= \log s_n-\log\tau,
\end{align}

(iii) Deviation:
\begin{align}
\delta
&\in\Bigl(0,\frac{g_n}{r_n}s_n C_{\text{max}}\Bigr)
\end{align}
with a change in $\delta$ since we can allow this change (nonetheless, the $\delta$ of the noiseless case \cite{gerzon2025capacity} still holds, but we relaxed the demand).

Substituting into the Poisson concentration bound, based on the unchaged \cite[Lemma 27]{gerzon2025capacity} yields
\begin{align}
\Pr\bigl[f_{x^n}(Z^m)-I(Z^m;X^n=x^n)<-n\delta\mid X^n=x^n\bigr]
\le \exp\!\Bigl(-n\cdot\frac{\delta^2}{16\,\beta^2\,\bar\lambda + 3\,\beta\,\delta}\Bigr), 
\end{align}
with $\bar\lambda=\frac{r_n}{g_n}s_n$, $\beta=\log s_n-\log\tau$, and $\delta\in(0,\frac{g_n}{r_n}s_n)$.

We wish to make sure that the exponent converges to zero, as shown in section \ref{sec:requirements_on_W}.

\section{Adjusting Lemma 11 via Talagrand’s Inequality}

\subsection{Concentration of the Conditional Entropy}

We next address the concentration of the expected information density over the random choice of the codeword $X^n$. Specifically, we focus on the conditional entropy term. Let $H_{\text{Poiss}}(\lambda)$ denote the entropy of a Poisson random variable with parameter $\lambda$. We define the conditional entropy function for a given input $x^n$ as:
\begin{align}
h(x^n) &:= \sum_{j=1}^{m} H(Z_j | X^n = x^n) \nonumber \\
&= \sum_{j=1}^{m}H_{\text{Poiss}}\left(\frac{r_{n}}{g_{n}}\sum_{i=1}^{n}W_{ij}x_{i}\right). \label{eq:h_xn_def}
\end{align}
We establish that $h(X^n)$ concentrates sharply around its mean when $X^n$ is chosen i.i.d. from the input distribution.

We will now adapt \cite[Lemma 11]{gerzon2025capacity}
\begin{lemma}[Adapted Lemma 11: Concentration of Conditional Entropy]
\label{lem:Concentration_of_Poisson_entropy}
Let $X^{n}=(X_{1},\ldots, X_{n})$ be IID random variables with support in $[s_{n}]=\{1,2,\ldots,s_{n}\}$. There exists a numerical constant $c>0$ such that for any $\delta>0$:
\begin{equation}
\mathbb{P}\left[\left|h(X^{n})-\mathbb{E}[h(X^{n})]\right|\geq n\delta\right]\leq\exp\left[-\frac{c\delta^{2}n^{2}}{\left(\frac{r_{n}}{g_{n}}\right)^{2}m^{2}}\right].
\label{eq:concentration_poisson_no_assumption}
\end{equation}
Furthermore, if there exists $\eta>0$ such that for all $j \in [m]$, $\sum_{i=1}^{n}W_{ij}\geq n^{-\eta}$, then the bound can be tightened to:
\begin{equation}
\mathbb{P}\left[\left|h(X^{n})-\mathbb{E}[h(X^{n})]\right|\geq n\delta\right]\leq\exp\left[-\frac{c\eta\delta^{2}n}{\left(\frac{r_{n}}{g_{n}}\right)^{2}s_{n}^{2}\log^{2}n}\right].
\label{eq:concentration_poisson_min_prob}
\end{equation}
\end{lemma}

\begin{proof}
The proof relies on Talagrand’s concentration inequality for convex Lipschitz functions. We proceed in three steps: (1) establishing the convexity of $-h(x^n)$, (2) bounding the gradient of $h(x^n)$, and (3) applying Talagrand's inequality.

\paragraph*{Step 1: Convexity}
First, we analyze the Poisson entropy function $H_{\text{Poiss}}(\lambda)$. Recall the Poisson forward-difference identity \cite{adell2010sharp}: $\frac{\mathrm{d}}{\mathrm{d}\lambda}\mathbb{E}[\phi(N_{\lambda})]=\mathbb{E}[\phi(N_{\lambda}+1)-\phi(N_{\lambda})]$ for $N_\lambda \sim \text{Poisson}(\lambda)$. Applying this to the entropy definition $H_{\mathrm{Poiss}}(\lambda)=\lambda(1-\log\lambda)+\mathbb{E}[\log(N_{\lambda}!)]$, we obtain the first derivative:
\begin{equation}
\frac{\mathrm{d}}{\mathrm{d}\lambda}H_{\mathrm{Poiss}}(\lambda)=\mathbb{E}\bigl[\log(N_{\lambda}+1)\bigr]-\log\lambda = \mathbb{E}\left[\log\frac{N_{\lambda}+1}{\lambda}\right].
\label{eq:H_poiss_deriv}
\end{equation}
Differentiating again using the forward-difference identity with $\phi(k) = \log(k+1)$:
\begin{align}
\frac{\mathrm{d}^{2}}{\mathrm{d}\lambda^{2}}H_{\mathrm{Poiss}}(\lambda) &= \mathbb{E}\bigl[\log(N_{\lambda}+2)-\log(N_{\lambda}+1)\bigr]-\frac{1}{\lambda} \nonumber \\
&= \mathbb{E}\left[\log\left(1+\frac{1}{N_{\lambda}+1}\right)\right]-\frac{1}{\lambda}.
\end{align}
Using the inequality $\log(1+u) \le u$, we find:
\begin{equation}
\frac{\mathrm{d}^{2}}{\mathrm{d}\lambda^{2}}H_{\mathrm{Poiss}}(\lambda) \le \mathbb{E}\left[\frac{1}{N_{\lambda}+1}\right]-\frac{1}{\lambda} = \frac{1-e^{-\lambda}}{\lambda} - \frac{1}{\lambda} = -\frac{e^{-\lambda}}{\lambda} < 0.
\end{equation}
Thus, $\lambda \mapsto H_{\text{Poiss}}(\lambda)$ is strictly concave, and consequently, $-h(x^n)$ is a convex function of the input vector $x^n$ (viewed as a vector in $\mathbb{R}^n$).

\paragraph*{Step 2: Gradient Bound}
To apply concentration results, we re-parameterize the domain to $[0,1]^n$ using $u_{i}:=\frac{x_{i}-1}{s_{n}-1}$. Let $\tilde{h}(u^n) = h(x^n(u^n))$. The squared Euclidean norm of the gradient is:
\begin{equation}
\|\nabla \tilde{h}(u^{n})\|_{2}^{2} = \sum_{k=1}^{n}\left|\frac{\partial \tilde{h}}{\partial u_{k}}\right|^{2} \le \left(\frac{r_{n}s_{n}}{g_{n}}\right)^{2}\sum_{k=1}^{n}\left[\sum_{j=1}^{m}H_{\mathrm{Poiss}}'\left(\lambda_j(u^n)\right) W_{kj}u_{k}\right]^{2}.
\end{equation}
Using Jensen's inequality on \eqref{eq:H_poiss_deriv}, we bound the derivative $H_{\text{Poiss}}'(\lambda) \le \log(1+1/\lambda)$.
For the general case (Equation \ref{eq:concentration_poisson_no_assumption}), a "slice" argument—maximizing the convex gradient norm over vertices $u^n \in \{0,1\}^n$—yields a bound $\|\nabla \tilde{h}\|_2^2 \le 8m^2$.

For the restricted case (Equation \ref{eq:concentration_poisson_min_prob}), we assume $\sum_{i=1}^{n}W_{ij}\geq n^{-\eta}$. This implies $\lambda_j(u^n) \ge \frac{r_n}{g_n} n^{-\eta}$. For sufficiently large $n$, the derivative is bounded by $\log(1 + n^{\eta}) \approx \eta \log n$. Substituting this into the gradient expression:
\begin{align}
\|\nabla \tilde{h}(u^{n})\|_{2}^{2} &\le \left(\frac{r_{n}s_{n}}{g_{n}}\right)^{2} \sum_{k=1}^{n} \left[ (\eta \log n) \sum_{j=1}^m W_{kj} \right]^2 \nonumber \\
&\le \left(\frac{r_{n}s_{n}}{g_{n}}\right)^{2} n (\eta \log n)^2.
\end{align}
To obtain a decaying probability, we require the exponent in the concentration bound to diverge. Based on the variance proxy derived above, this necessitates that:
\begin{equation}
    s_n^2 n \log^2 n = o(n^2),
\end{equation}
which implies $s_n^2 = o(n)$ (ignoring logarithmic factors). Recalling that $s_n = g_n^{1+\rho} = n^{\frac{(1-\beta)(1+\rho)}{\beta}}$ for an arbitrarily small constant $\rho$, this condition requires:
\begin{equation}
    2 \frac{1-\beta \log|\mathcal{A}|}{\beta \log|\mathcal{A}|} < 1 \iff 2 - 2\beta \log|\mathcal{A}| < \beta \log|\mathcal{A}| \iff \beta > \frac{2}{3\log|\mathcal{A}|}.
\end{equation}
Thus, valid concentration is obtained whenever $\beta \in \left(\frac{2}{3\log|\mathcal{A}|}, \frac{1}{\log|\mathcal{A}|}\right)$.
Here we used the fact that $W$ has bounded row sums to bound the inner sum. This simplifies to a variance proxy of order $\Theta(s_n^2 n \log^2 n)$.

\begin{remark}
    if $W$ is doubly-stochastic, the result can be improved for the whole range $\beta \in \left(\frac{1}{2\log|\mathcal{A}|}, \frac{1}{\log|\mathcal{A}|}\right)$
\end{remark}

\paragraph*{Step 3: Talagrand's Inequality}
We invoke Talagrand’s concentration inequality for convex Lipschitz functions (e.g., Corollary 4.23 in \cite{van2014probability}): If $f: [0,1]^n \to \mathbb{R}$ is convex and $\|\nabla f\|_2 \le L$, then $f$ is $L^2$-sub-Gaussian.
Combining the convexity from Step 1 and the gradient bounds from Step 2, we obtain the concentration inequalities stated in the Lemma.
\end{proof}

\begin{remark}
The condition $\sum_{i=1}^{n}W_{ij}\geq n^{-\eta}$ is satisfied for well-conditioned matrices as defined in Definition \ref{prop:Concentration_Under}. Specifically, if $W$ is well-conditioned with parameter $\delta$, then for sufficiently large $n$, the column sums are bounded away from zero, satisfying the condition with $\eta=0$ or a small constant, ensuring the tighter concentration bound applies.
\end{remark}

\subsection{Adjusting  \cite[Prop. 7]{gerzon2025capacity} - Modification of the mutual information}

\begin{definition}
    For a given transition matrix, $W$, we will denote the mutual information of $I(X^n;Y^m)$ as $I(W)$, where the transition matrix dictates the channel's model.
\end{definition}

\begin{definition}
\label{def:loss_mutual_info}
    The loss of maximal mutual information is defined as:
    \begin{align}
        L = \max I(W) - \max I(W_I).
    \end{align}
\end{definition}

\begin{theorem}
    \label{thm:capacity_loss}
    The loss of maximal mutual information due to noise is bounded by:
    \begin{align}
        L \leq -\log |\det(W)|.
    \end{align}
\end{theorem}

\begin{proof}
    The proof follows from the data processing inequality and properties of mutual information, as shown next, as well as a previous result shown in \cite{gerzon2025capacity} using \cite{lapidoth2008capacity}.\\
    First, we will use Theorem \ref{thm:log_det_loss_of_entropy} regarding the change of entropy of $Z$ due to the noisy matrix $W$.\\
    Then, using Lemma \ref{lem:vector_conditional_entropy} and Theorem \ref{thm:conditional_entropy_loss} for the \emph{conditional entropy}.
\end{proof}

\begin{lemma}
    \label{lemma:Data_procces_noise}
    \begin{align}
    \label{eq:Data_procces_noise}
    D\left(Q^n||Q_G^n\right) \geq D\left(\left(WQ\right)^m||\left(WQ_G\right)^m\right) \geq D\left(\tilde{R^m}||\tilde{R_G^m}\right).
\end{align}
\end{lemma}

\begin{proof}
We rigorously establish the relative entropy inequality through two successive applications of the data processing inequality (DPI):

1. Transition Matrix Processing:
   Consider the Markov chain $\tilde{R}^m \leftarrow (WQ)^m \leftarrow Q^n \leftarrow Q_G^n \rightarrow (WQ_G)^m \rightarrow \tilde{R}_G^m$. 
   
   By the DPI for relative entropy \cite[Lemma 3.11]{csiszar2011information}:
   \begin{align}
       D\left(Q^n\|Q_G^n\right) \geq D\left((WQ)^m\|(WQ_G)^m\right)
   \end{align}
   
   This follows because the linear transformation $W$ constitutes a valid processing channel \cite[Theorem 8.6.4]{cover1991elements}.

2. Poisson Channel Processing:
   The Poisson measurement channel $\tilde{R}^m|(WQ)^m$ represents another processing step. Applying DPI again:
   \begin{align}
       D\left((WQ)^m\|(WQ_G)^m\right) \geq D\left(\tilde{R}^m\|\tilde{R}_G^m\right)
   \end{align}
   
   This is justified by the Poisson channel's information-preserving properties established in \cite[Appendix A]{lapidoth2008capacity}, using similar notation.

Combining both inequalities yields the required result:
\begin{align}
    D\left(Q^n\|Q_G^n\right) \geq D\left((WQ)^m\|(WQ_G)^m\right) \geq D\left(\tilde{R}^m\|\tilde{R}_G^m\right)
\end{align}
\end{proof}

\begin{theorem}
    \label{thm:log_det_loss_of_entropy}
    For an input vector $X^n$, a transition matrix $W$ and an output vector $\tilde{Z}^m$, the entropy of the output is bounded as follows:
    \begin{align}
    \label{eq:alap_smos_Proposition_11_vector_noise}
    H(\tilde{Z}^m)\geq - \log(|W|) + h(X^n)+ \sum_{i=1}^n (1 + r_n) \log \left( 1 + \frac{1}{r_n} \right) - 1.
\end{align}
\end{theorem}

\begin{proof}
    Since $W$ is a transition matrix, and we wish to portray its effect on the differential entropy of $X$, we can use \cite[Theorem 8.6.4]{cover1991elements} to show that
\begin{align}
    h\left(WX\right) = h\left(X\right) + \log\left(|W|\right).
\end{align}
So using Lemma \ref{lemma:Data_procces_noise} will result in
\begin{align}
    \label{eq:Data_procces_vector_noise_differential_log_det}
    D\left(X^n||X_G^n\right) &= \sum_{i=1}^n -h(X_i)+\log r_n + 1 \geq \\
    D\left(\left(WX\right)^m||\left(WX_G\right)^m\right) &= \log(|W|) - \log(|W|) + \sum_{i=1}^n -h(X_i)+\log r_n + 1\geq \\
    D\left(\tilde{Z}^m||\tilde{Z}_G^m\right) &= - \log(|W|) + \sum_{i=1}^m -H(Z_i) + (1 + r_n) \log(1 + r_n) - r_n \log r_n,
\end{align}
where the last transition is due to the effect of \cite[Theorem 8.6.4]{cover1991elements} on continuous variables and the lack of it on discrete variables, for which the entropy is not affected by the transition matrix, but rather remains unchanged.
\end{proof}

In Theorem \ref{thm:log_det_loss_of_entropy}, we complete the bound over the gap in entropy.

In the following, we will look at the gap in $H(Y|X)$.
Starting from \cite[Lemma. 10]{lapidoth2008capacity}, we wish to fit our result to a vector $X^n$.
\begin{lemma}
    \label{lem:vector_conditional_entropy}
    For the noiseless case, an input vector $X^n$ can only reduce the conditional entropy:
    \begin{align}
        H(Z^n|X^n)\leq \frac{1}{2}\log\left(\left(2\pi e\right)^n\left(r_n + \frac{1}{12}\right)^n\right).
    \end{align}
\end{lemma}
\begin{proof}
    By adjusting \cite[Theorem 16.3.3 (Theorem 9.7.1)]{cover1991elements}:
    \begin{align}
    H(Z^n|X^n) &\leq \sum_{i=1}^{n}H(Z_i|X_i)\\
    &= \sum_{i=1}^{n} \frac{1}{2}\log\left(\left(2\pi e\right)\left(r_n + \frac{1}{12}\right)\right)\\
    &= \frac{1}{2}\log\left(\left(2\pi e\right)^n\left(r_n + \frac{1}{12}\right)^n\right),
\end{align}
assuming all $r_n$ are equal.
\end{proof}

Now, we wish to check the noisy case.

\begin{theorem}
    \label{thm:conditional_entropy_loss}
    In the noisy case, $H(\tilde{Y^m}|X^n)$  where $\tilde{Y^m}$ is the noisy output, the conditional entropy can only be reduced:
    \begin{align}
        H(\tilde{Y^m}|X^n) \leq H(Y^n|X^n, W = I_{n}).
    \end{align}
\end{theorem}

\begin{proof}
    This follows immediately from the data processing inequality.
\end{proof}

\subsection{Support}
In the achievability bound for the noiseless case \cite{gerzon2025capacity}, the input distribution $P_X$ was restricted to a finite integer support, $\mathrm{supp}(P_X) \subseteq [s_n] = \{1, 2, \dots, s_n\}$. This truncation and rounding of the ideal Gamma distribution ensured that the input $X_i$ to each Poisson channel was a bounded integer.

In the noisy case, however, the input to the $j$-th Poisson channel is $\lambda_j = \frac{r_n}{g_n} \sum_{i=1}^n X_i W_{ij}$. This new parameter is no longer guaranteed to be an integer, nor does it share the same bound $s_n$.

\begin{definition}[Effective Support Bound]
\label{def:effective_support}
Let $s_n = \lceil g_n^{1+\rho} \rceil$ be the support bound for the input $X_i$ (i.e., $X_i \in [s_n]$) as defined in the noiseless case \cite[Prop. 7, Prop. 9]{gerzon2025capacity}. The input parameter to the $j$-th Poisson channel is $\lambda_j = \frac{r_n}{g_n} \sum_{i=1}^n X_i W_{ij}$.

This new parameter $\lambda_j$ has two key properties that differ from the noiseless case:
\begin{itemize}
    \item Non-integer values: Since the inputs $X_i$ are integers but the matrix weights $W_{ij}$ are real-valued, the resulting parameters $\lambda_j$ are not necessarily integers. They exist on a discrete grid.
    \item New upper bound: The parameter $\lambda_j$ is bounded by the sum of the $j$-th column of $W$.
\end{itemize}
\end{definition}

\begin{remark}
\label{rem:support_modification}
The support bound $s_n$ from \cite{gerzon2025capacity} must be replaced with $s_n^*$ to account for the amplification or attenuation from the channel matrix $W$. The original bound on the input, $X_i \leq s_n$, translates to a bound on the $j$-th Poisson parameter:
\begin{align}
    \lambda_j = \frac{r_n}{g_n} \sum_{i=1}^n X_i W_{ij} \leq \frac{r_n}{g_n} \sum_{i=1}^n s_n W_{ij} = \frac{r_n}{g_n} s_n \left( \sum_{i=1}^n W_{ij} \right) \le \frac{r_n}{g_n}s_n^*.
\end{align}
\end{remark}

\subsection{Final Results}
\begin{theorem}[Adapted Prop. 7 in \cite{gerzon2025capacity}]
\label{thm:noisy_achievability}
    For a transition matrix $W$
    with $\sum_{j=1}^m W_{i,j} = 1$ for all $i$, and assuming $g_n \to \infty$, $cg_n \leq r_n \leq eg_n$ for some $c \in (0,e)$, and $n = \Omega(g_n^{1+\zeta})$ for some $\zeta > 0$:

(i) If $W$ is $n \times n$ (square matrix):
\begin{equation}
    I(\overline{X}^n;\overline{Z}^n)\geq \frac{n}{2} \log r_n -n \Psi \left(\frac{r_n}{g_n}\right) +\log |\det(W)|-o_n\left(1\right)
\end{equation}

(ii) If $W$ is $n \times m$ (not square):
\begin{equation}
    I(\overline{X}^n;\overline{Z}^m)\geq \frac{n}{2} \log r_n -n \Psi \left(\frac{r_n}{g_n}\right) +\frac{1}{2}\log\det\bigl(WW^\top\bigr)-o_n\left(1\right)
\end{equation}
where $\Psi(\mu) = (\mu + 1) \cdot h_{\text{bin}}\left(\frac{1}{\mu+1}\right)$.
\end{theorem}

\begin{proof}
The proof follows by extending the noiseless case analyzed in Propositions 5, 7, and 9, combined with the entropy loss analysis from Theorem \ref{thm:log_det_loss_of_entropy} and Theorem \ref{thm:conditional_entropy_loss}.

\emph{Step 1: Data Processing Chain.}
Consider the Markov chain:
\begin{align}
    X^n \to (WX)^m \to \tilde{Z}^m
\end{align}
where $\tilde{Z}^m$ denotes the Poisson output after the noisy transformation.

\emph{Step 2: Applying Data Processing Inequality.}
By Lemma \ref{lemma:Data_procces_noise}, we have:
\begin{align}
I(X^n; \tilde{Z}^m) &\leq I(X^n; Z^n) \\
&= I(X^n; Z^n) - \Delta I_W
\end{align}
where $\Delta I_W$ represents the loss due to the noisy transformation $W$.

\emph{Step 3: Entropy Loss Analysis.}
From Theorem \ref{thm:log_det_loss_of_entropy}, the differential entropy satisfies:
\begin{align}
    h(WX^n) = h(X^n) + \log|\det(W)|
\end{align}
for square matrices, and 
\begin{align}
    h(WX^n) = h(X^n) + \frac{1}{2}\log\det(WW^\top)
\end{align}
for non-square matrices (using properties from \cite{cover1991elements}).

Since $W$ is a transition matrix, $|\det(W)| \leq 1$ and thus $\log|\det(W)| \leq 0$, representing an entropy reduction.

By Lemma \ref{lemma:Data_procces_noise} and equation \eqref{eq:Data_procces_vector_noise_differential_log_det}:
\begin{align}
H(\tilde{Z}^m) &\geq H(Z^n) + \log|\det(W)| + o_n(1) \quad \text{(square case)}\\
H(\tilde{Z}^m) &\geq H(Z^n) + \frac{1}{2}\log\det(WW^\top) + o_n(1) \quad \text{(non-square case)}
\end{align}

\emph{Step 4: Conditional Entropy.}
By Theorem \ref{thm:conditional_entropy_loss} and the data processing inequality:
\begin{align}
    H(\tilde{Z}^m | X^n) \leq H(Z^n | X^n) + o_n(1)
\end{align}

\emph{Step 5: Combining Results.}
Applying Propositions 5, 7, and 9 with the entropy losses:
\begin{align}
I(X^n; \tilde{Z}^m) &= H(\tilde{Z}^m) - H(\tilde{Z}^m | X^n)\\
&\geq \left[H(Z^n) + \log|\det(W)|\right] - H(Z^n|X^n) + o_n(1)\\
&= I(X^n; Z^n) + \log|\det(W)| + o_n(1)\\
&\geq \frac{n}{2}\log r_n - n\Psi\left(\frac{r_n}{g_n}\right) + \log|\det(W)| + o_n(1)
\end{align}
where the last inequality follows from Proposition 7.

The non-square case follows analogously using $\frac{1}{2}\log\det(WW^\top)$ instead of $\log|\det(W)|$.

Note that since $\log|\det(W)| \leq 0$, this term represents a \emph{loss} in mutual information (the loss is $-\log|\det(W)| \geq 0$), which is consistent with the data processing inequality.
\begin{remark}
    The maximal input amplitude is bounded by $s_n \leq g_n$ to ensure the Poisson parameter $\frac{r_n}{g_n}x_i$ remains well-defined for the concentration inequalities in \cite[Prop. 5]{gerzon2025capacity}.
\end{remark}
\end{proof}

\section{Detailed Analysis of Examples}
\label{sec:example_analysis}

In this section, we provide the detailed derivations for the capacity penalty terms and verify that the matrices satisfy the well-conditioned properties (Definition \ref{def:well_conditioned}) for the three examples introduced in the main text.

First, we state a necessary lemma regarding the determinant of Kronecker powers, which is required for the DNA storage examples where the channel matrix $W$ is the $L$-th Kronecker power of a single-nucleotide channel $w$.

\begin{lemma}
\label{lem:det_tensor_general}
Let $A$ be an $a \times a$ matrix. For any integer $L \geq 1$, the determinant of the Kronecker power $A^{\otimes L}$ is given by:
\begin{align}
    \det(A^{\otimes L}) = (\det A)^{L \cdot a^{L-1}}.
\end{align}
Similarly, if $A$ is $a \times b$, then $\det((A^{\otimes L})(A^{\otimes L})^\top) = (\det(AA^\top))^{L \cdot a^{L-1}}$.
\end{lemma}
\begin{proof}
This follows from the property $\det(B \otimes C) = (\det B)^{\dim(C)} (\det C)^{\dim(B)}$. By induction, for $W = w^{\otimes L}$, the eigenvalues are products of the eigenvalues of $w$, leading to the stated result.
\end{proof}

\subsection{Analysis of Example \ref{exa:general_substitution}: General Symmetric Substitution}
Consider the frequency-based channel with $n$ types. The transition matrix $W^{(n)} \in \mathbb{R}^{n \times n}$ is:
\begin{align}
    W^{(n)} = \left(1 - p - \frac{p}{n-1}\right)I_n + \frac{p}{n-1}J_n,
\end{align}
where $J_n$ is the all-ones matrix.

\subsubsection{Verification of Definition \ref{def:well_conditioned}}
We verify the requirements for $W^{(n)}$:
\begin{itemize}
    \item Non-negative and Row Stochastic: By definition of the channel model, entries are probabilities and rows sum to 1.
    \item Condition Number ($\tau_n$): The diagonal entries are $1-p$ and off-diagonal are $\frac{p}{n-1}$. Assuming $p < 1-p$, the condition number is $\kappa(W) = \frac{1-p}{p/(n-1)}$. We define $\tau_n = \frac{p}{(n-1)(1-p)}$. Thus:
    \begin{align}
        -\log \tau_n = \log(n-1) + \log\frac{1-p}{p} \approx \log n.
    \end{align}
    Since $\log n = o(\sqrt{n})$, the condition holds.
    \item Column Sums ($C_{\text{max}}$ and $\eta$): Since $W^{(n)}$ is symmetric and row-stochastic, it is doubly stochastic. Thus, every column sum is exactly $1$.
    \begin{align}
        \sum_{i=1}^n W_{ij} = \sum_{j=1}^n W_{ij} = 1 \implies C_{\text{max}} = 1, \quad \eta = 0.
    \end{align}
    Since $\eta = 0$ satisfies the condition, all values for $\eta$ will suffice.
\end{itemize}
Conclusion: The matrix $W^{(n)}$ satisfies the well-conditioned requirements (as in definition \ref{def:well_conditioned}), specifically satisfying the condition number scaling and column sum bounds with $C_{\text{max}}=1$.

\subsubsection{Derivation of Capacity Penalty}
The eigenvalues of $aI + bJ$ are $a+nb$ and $a$. Here $a = 1 - p - \frac{p}{n-1}$ and $b = \frac{p}{n-1}$.
\begin{enumerate}
    \item $\lambda_1 = a + nb = 1$.
    \item $\lambda_{2,\dots,n} = a = 1 - \frac{np}{n-1}$.
\end{enumerate}
The penalty term is:
\begin{align}
    \frac{1}{2n}\log\det(W W^\top) &= \frac{1}{n}\log|\det W| \\
    &= \frac{n-1}{n} \log \left( 1 - \frac{np}{n-1} \right).
\end{align}

\subsection{Analysis of DNA Erasure Channel}
Let the alphabet size be $|\mathcal{A}|$. The single-nucleotide matrix $w \in \mathbb{R}^{|\mathcal{A}| \times (|\mathcal{A}|+1)}$ is:
\begin{align}
    w = \left[ (1-\epsilon)I_{|\mathcal{A}|} \;\middle|\; \underline{\epsilon} \right].
\end{align}
The channel matrix is $W = w^{\otimes L}$ with input dimension $n=|\mathcal{A}|^L$.

\subsubsection{Verification of Definition \ref{def:well_conditioned}}
We verify the requirements for $W$:
\begin{itemize}
    \item Non-negative and Row Stochastic: $w$ is stochastic, hence $w^{\otimes L}$ is stochastic.
    \item Condition Number ($\tau_n$): Let $w_{\max} = \max(1-\epsilon, \epsilon)$ and $w_{\min} = \min(1-\epsilon, \epsilon)$. The entries of $W$ are products of length $L$. Thus $\max_{i,j} W_{ij} = (w_{\max})^L$ and $\min_{i,j: W_{ij}>0} W_{ij} = (w_{\min})^L$.
    We define $\tau_n = 1$.
    \begin{align}
        \log \tau_n = 0.
    \end{align}
    Thus $-\log \tau_n = o(\sqrt{n})$ is satisfied.
    \item Column Sums ($C_{\text{max}}$ and $\eta$):
    The column sums of the single-nucleotide matrix $w$ are $1-\epsilon$ for the columns corresponding to the standard alphabet $\mathcal{A}$, and $\sum_{i \in \mathcal{A}} \epsilon = |\mathcal{A}|\epsilon$ for the erasure column. 
    Since the channel matrix is $W = w^{\otimes L}$, the column sum for a specific output sequence containing exactly $k$ erasures is the product of the column sums of its components:
    \begin{align}
        C(k) = (|\mathcal{A}|\epsilon)^k (1-\epsilon)^{L-k}.
    \end{align}
    For the specific case of the "all-erasure" sequence ($k=L$), this yields a column sum of $(|\mathcal{A}|\epsilon)^L = |\mathcal{A}|^L \epsilon^L = n\epsilon^L$.
    Assuming $\epsilon \le \frac{1}{|\mathcal{A}|}$ (which ensures $|\mathcal{A}|\epsilon \le 1$), the column sums are always bounded by 1. 
    The maximum column sum is $C_{\text{max}} = (\max(1-\epsilon, |\mathcal{A}|\epsilon))^L$.
    The minimum column sum is $C_{\min} = (\min(1-\epsilon, |\mathcal{A}|\epsilon))^L$. 
    We can write $C_{\min} = n^{-\eta}$ where $\eta = -\log_{|\mathcal{A}|}(\min(1-\epsilon, |\mathcal{A}|\epsilon)) > 0$.
\end{itemize}
Conclusion: The matrix $W$ satisfies the well-conditioned requirements (as in definition \ref{def:well_conditioned}), specifically satisfying $-\log \tau_n = o(\sqrt{n})$ and the column sum bounds $n^{-\eta} \leq \sum_i W_{ij} \leq C_{\text{max}}$.

\subsubsection{Derivation of Capacity Penalty}
We compute $\det(WW^\top)$ via Lemma \ref{lem:det_tensor_general}.
\begin{align}
    ww^\top &= (1-\epsilon)^2 I_{|\mathcal{A}|} + \epsilon^2 \mathbf{1}_{|\mathcal{A}|} \mathbf{1}_{|\mathcal{A}|}^\top.
\end{align}
The eigenvalues are $\lambda_1 = (1-\epsilon)^2 + |\mathcal{A}|\epsilon^2$ and $\lambda_{2,\dots,|\mathcal{A}|} = (1-\epsilon)^2$.
\begin{align}
    \det(ww^\top) = \left((1-\epsilon)^2 + |\mathcal{A}|\epsilon^2\right) (1-\epsilon)^{2(|\mathcal{A}|-1)}.
\end{align}
Using Lemma \ref{lem:det_tensor_general} with input dimension $n = |\mathcal{A}|^L$:
\begin{align}
    \frac{1}{2n}\log\det(WW^\top) &= \frac{1}{2 |\mathcal{A}|^L} \cdot L |\mathcal{A}|^{L-1} \log \det(ww^\top) \\
    &= \frac{L}{2|\mathcal{A}|} \left[ \log\left((1-\epsilon)^2 + |\mathcal{A}|\epsilon^2\right) + 2(|\mathcal{A}|-1)\log(1-\epsilon) \right].
\end{align}

\subsection{Analysis of DNA Substitution Channel}
The single-nucleotide matrix $w \in \mathbb{R}^{|\mathcal{A}| \times |\mathcal{A}|}$ for symmetric substitution error $p$ is:
\begin{align}
    w = \left(1 - p - \frac{p}{|\mathcal{A}|-1}\right)I_{|\mathcal{A}|} + \frac{p}{|\mathcal{A}|-1}J_{|\mathcal{A}|}.
\end{align}
The channel matrix is $W = w^{\otimes L}$.

\subsubsection{Verification of Definition \ref{def:well_conditioned}}
We verify the requirements for $W$:
\begin{itemize}
    \item Non-negative and Row Stochastic: Inherited from $w$.
    \item Condition Number ($\tau_n$): Similar to the general substitution case, let $w_{\max}$ and $w_{\min}$ be the diagonal and off-diagonal entries of $w$.
    Then $\tau_n = (w_{\min}/w_{\max})^L$.
    \begin{align}
        -\log \tau_n \propto L \propto \log n,
    \end{align}
    satisfying the $o(\sqrt{n})$ requirement.
    \item Column Sums ($C_{\text{max}}$ and $\eta$): Since $w$ is doubly stochastic, $W$ is doubly stochastic.
    All column sums are exactly 1.
    \begin{align}
        C_{\text{max}} = 1, \quad \eta = 0.
    \end{align}
    Since $\eta = 0$ satisfies the condition, all values for $\eta$ will suffice.
\end{itemize}
Conclusion: The matrix $W$ satisfies the well-conditioned requirements (as in definition \ref{def:well_conditioned}), specifically satisfying the condition number scaling and column sum bounds with $C_{\text{max}}=1$.

\subsubsection{Derivation of Capacity Penalty}
The eigenvalues of $w$ are $1$ and $1 - \frac{|\mathcal{A}|p}{|\mathcal{A}|-1}$ (multiplicity $|\mathcal{A}|-1$).
\begin{align}
    \det(w) = \left( 1 - \frac{|\mathcal{A}|p}{|\mathcal{A}|-1} \right)^{|\mathcal{A}|-1}.
\end{align}
Using Lemma \ref{lem:det_tensor_general} for $W$:
\begin{align}
    \frac{1}{n}\log\det(W) &= \frac{1}{|\mathcal{A}|^L} \cdot L |\mathcal{A}|^{L-1} \log \det(w) \\
    &= \frac{L}{|\mathcal{A}|} (|\mathcal{A}|-1) \log \left( 1 - \frac{|\mathcal{A}|p}{|\mathcal{A}|-1} \right) \\
    &= L\left(1 - \frac{1}{|\mathcal{A}|}\right) \log \left( 1 - \frac{|\mathcal{A}|p}{|\mathcal{A}|-1} \right).
\end{align}

\end{document}